\documentclass[pra,twocolumn,showpacs,preprintnumbers,amsmath,amssymb]{revtex4}
\usepackage{graphicx}
\usepackage{dcolumn}
\usepackage{bm}
\usepackage{amsthm}
\usepackage{amscd}
\usepackage{epsfig}
\usepackage{color}
\input xy
\xyoption{all}

\def\wt{\mbox{\rm wt}}
\newcommand{\abs}[1]{\left| #1 \right|}
\newcommand{\ket}[1]{\left| #1 \right\rangle}
\newcommand{\bra}[1]{\left\langle #1 \right |}
\newtheorem{theorem}{Theorem}
\newtheorem{lemma}{Lemma}
\newtheorem{proposition}{Proposition}

\DeclareMathOperator{\tr}{tr}

\DeclareMathOperator{\Stab}{Stab}

\DeclareMathOperator{\Sym}{Sym}

\newcommand{\R}{{\mathbb R}}

\newcommand{\C}{{\mathbb C}}

\newcommand{\Id}{{\rm Id}}
\newcommand{\GHZ}{{\rm G}}
\newcommand{\kghz}{K_{\rm \small GHZ}}
\newcommand{\kwerner}{K_{\rm \small Werner}}
\newcommand{\kproduct}{K_{\rm \small product}}
\newcommand{\kdicke}{K_{\rm \small Dicke}}

\newcommand{\twotwo}[4]{\left[ \begin{array}{cc} #1 & #2 \\
                        #3 & #4 \end{array} \right]}

\begin{document}

\title{Symmetric mixed states of $n$ qubits: local unitary stabilizers and
  entanglement classes}

\author{David W. Lyons}
  \email{lyons@lvc.edu}
\author{Scott N. Walck}
  \email{walck@lvc.edu}
\affiliation{Lebanon Valley College, Annville, PA 17003}

\date{7 July 2011, revised 4 October 2011}

\begin{abstract}
We classify, up to local unitary equivalence, local unitary stabilizer
Lie algebras for symmetric mixed states of $n$ qubits into six
classes. These include the stabilizer types of the Werner states, the
GHZ state and its generalizations, and Dicke states. For all but the
zero algebra, we classify entanglement types (local unitary equivalence
classes) of symmetric mixed states that have those stabilizers. We make
use of the identification of symmetric density matrices with polynomials
in three variables with real coefficients and apply the representation
theory of $SO(3)$ on this space of polynomials.
\end{abstract}


\maketitle

\section{Introduction}\label{introsection}

Quantum information seeks to exploit quantum states as resources for
computational tasks, where the phenomenon of entanglement seems to play
a central role in quantum advantages over classical protocols. A problem
at the heart of the study of entanglement is to understand the essential
differences between local and global unitary operations on composite
systems. Questions such as the following arise naturally. How can one
characterize the nonlocal properties of a given state? Given an input
state, what is a reasonable description of the set of states to which it
can be converted using only local operations?  These questions are known
to be hard. Nonlocal properties are poorly
understood~\cite{RevModPhys.81.865}. Regarding the second question, the
number of polynomial invariants needed to determine local unitary
equivalence classes of multiqubit states grows exponentially with the
number of qubits~\cite{meyer01p}.

The considerations of the previous paragraph motivate the search for
entanglement measures and local unitary invariants that, being more
tractable than a complete set of polynomial invariants, will provide
a coarser grained classification, yet still fine enough to distinguish
useful states and good entanglement properties. We pursue this
philosophy with the following scheme. Given an $n$-qubit density matrix
$\rho$, its local unitary stabilizer subgroup is defined to be all the
local transformations of state space that leave $\rho$ invariant. If
$\rho,\rho'$ are local unitary equivalent, say by a local unitary
operation $U$, then their stabilizers are isomorphic via conjugation by
$U$~\cite{symmstates2}. Thus the conjugacy class of the stabilizer
subgroup of a state is a local unitary invariant of that state. This
invariant is practical and has proven to distinguish states with
known useful entanglement properties. Tractability is achieved by virtue
of the Lie group structure of the local unitary group and stabilizer
subgroups. Lie groups admit analysis through their Lie algebras, which
are their tangent spaces of infinitesimal local operations. In previous work, we have
classified conjugacy classes of stabilizers and local unitary classes of
states including the singlet, GHZ and its generalizations, symmetric
Dicke states (including the W state and its generalizations), and Werner
states~\cite{maxstabnonprod1,maxstabnonprod2,su2blockstates,symmstatespaper,symmwerner}.

In this article, we consider the class of symmetric mixed states of
$n$-qubit systems. Symmetric states, that is, states that are invariant
under permutation of subsystems, are the subject of recent work
including: geometric measure of entanglement
\cite{aulbach2010,aulbach2010b,markham2010,PhysRevA.82.032301,1367-2630-12-8-083002},
efficient tomography \cite{toth2010}, classification of states
equivalent under stochastic local operations and classical communication
(SLOCC) \cite{bastin2009,bastin2010,PhysRevLett.105.200501}, and our own
work on classification of pure symmetric states equivalent under local unitary (LU)
transformations \cite{symmstatespaper,symmstates2}. 

The main results of this paper are: a classification of the six local
unitary stabilizer subalgebras (Lie algebras of the local unitary
stabilizer subgroups) in Theorem~\ref{mixedsymmkrhodecomp}; and for each
of those algebras (except for the zero algebra), a classification of
local unitarily distinct classes of states. In addition,
Theorem~\ref{lualgdecomp} gives a structure theorem for
stabilizer subalgebras of mixed states that generalizes a similar result
for pure states in our earlier work~\cite{maxstabnonprod1}.

The paper is organized as follows. We establish notation and convention
in Section~\ref{prelimsection}. Section~\ref{symmstate2poly} presents
the identification of symmetric mixed states with polynomials in 3
variables and the basics of $SO(3)$ representation theory (already given
in our previous work~\cite{symmwerner} but reproduced here for the sake
of self containment) needed for the analysis of stabilizers and their
corresponding states. We give structure theorems for subalgebras of the
local unitary algebra in general (Theorem~\ref{lualgdecomp}), and
stabilizer algebras of symmetric mixed states in particular
(Theorem~\ref{mixedsymmkrhodecomp}), in
Section~\ref{algstructsection}. Proofs of these theorems are given in
the appendix. Five subsequent sections analyze and classify entanglement
types of symmetric mixed states (local unitary equivalence classes) for
five of the six stabilizer types (all but the zero algebra) given in
Theorem~\ref{mixedsymmkrhodecomp}.

\section{Preliminaries}\label{prelimsection}

Let $G=(SU(2))^n$ denote the $n$-qubit local unitary group, where
$g=(g_1,g_2,\ldots,g_n) \in G$ acts on a density matrix $\rho$ by
$$ g\rho g^\dagger := (g_1\otimes g_2\otimes \cdots \otimes g_n) \rho
(g_1\otimes g_2\otimes \cdots \otimes g_n)^\dagger. 
$$ Let $LG=\bigoplus_{i=1}^n su(2)$ denote the local unitary Lie algebra
of infinitesimal transformations, where $su(2)$ is the set of $2\times
2$ skew-Hermitian matrices with trace zero.  In order to study the
actions of $G$ and $LG$ on density matrices, it is convenient to
consider a larger vector space that contains the set of density matrices
as a proper subset.  Let ${\cal W}$ denote the 4-dimensional real
vector space of $2 \times 2$ Hermitian matrices.  A convenient basis for
${\cal W}$ is $\{ \sigma_0, \sigma_1, \sigma_2, \sigma_3 \}$, where $\sigma_0$
is the $2 \times 2$ identity matrix, and $\sigma_1=\sigma_x$,
$\sigma_2=\sigma_y$, and $\sigma_3=\sigma_z$ are the Pauli matrices.

The set of $n$-qubit density matrices is a proper subset of the vector space ${\cal
  W}^{\otimes n}$, where every element $\rho$ (whether or not $\rho$ is
positive or has trace 1) can be uniquely written in the form $\rho=\sum_I
s_I\sigma_I$, where $I=i_1i_2\ldots i_n$ is a multiindex with
$i_k=0,1,2,3$ for $1\leq k\leq n$, and $\sigma_I$ denotes
$$\sigma_I = \sigma_{i_1}\otimes \sigma_{i_2}\otimes \cdots \otimes
\sigma_{i_n},$$
and the coefficients $s_I$ are real. 
When it is necessary to express a density matrix in terms of the
computational basis, we use the notation $\rho=\sum_{I,J}\rho_{I,J}\ket{I}\bra{J}$,
where $I=i_1i_2\ldots i_n$ and $J=j_1j_2\ldots j_n$ are bit strings of
length $n$. For $I=i_1i_2\cdots i_n$, we write $I_k$ to denote the
string $i_1i_2\cdots i_k^c \cdots i_n$ obtained by complementing the
$k$th index, where $i_k^c = i_k + 1 \pmod 2$. Note that the same multiindex notation $I=i_1i_2\ldots i_n$ means
$i_k=0,1,2,3$ for the Pauli tensor expansion, and means $i_k=0,1$ for the
computational basis expansion, where context makes the meaning clear.

An element $g\in SU(2)$ acts~\footnote{When we identify the element
  $a\sigma_x+b\sigma_y+c\sigma_z$ in ${\cal W}$ with
  $i(a\sigma_x+b\sigma_y+c\sigma_z)$ in $su(2)$,
  equation~(\ref{Adaction}) is the adjoint action of $SU(2)$ on its Lie
  algebra $su(2)$, plus a trivial action on the real linear span of
  $\sigma_0$. Equation~(\ref{adaction}) is the corresponding adjoint
  action of the Lie algebra $su(2)$ on itself, plus again the trivial
  action on $\sigma_0$. Equations~(\ref{Adactionext})
  and~(\ref{adactionext}) are the natural extensions of~(\ref{Adaction})
  and~(\ref{adaction}) to tensor products. See, for example,
  \cite{b-td:rclg}.}  on $\sigma_i$ by
\begin{equation}\label{Adaction}
\sigma_i\mapsto g\sigma_i g^\dagger
\end{equation}
and the corresponding action of $M$ in $su(2)$ on $\sigma_i$ is 
\begin{equation}\label{adaction}
\sigma_i\mapsto [M,\sigma_i]=M\sigma_i - \sigma_iM.
\end{equation}
An element $g=(g_1,g_2,\ldots,g_n)$ in
$G$ acts on $\sigma_I$ by
\begin{equation}\label{Adactionext}
g\sigma_I g^\dagger := (g_1\sigma_{i_1}g_1^\dagger)\otimes \cdots
\otimes (g_n\sigma_{i_n}g_n^\dagger)
\end{equation}
and $M=(M_1,M_2,\ldots,M_n)$ in $LG$ acts on $\sigma_I$ by 
\begin{equation}\label{adactionext}
  [M,\sigma_I] := \sum_{k=1}^n \sigma_{i_1}\otimes  \cdots \otimes [M_k,\sigma_{i_k}]\otimes\cdots
\otimes \sigma_{i_n}
\end{equation}

The local unitary stabilizer of $\rho\in {\cal W}^{\otimes n}$
is the subgroup $\Stab_\rho$ of $G$ given by
$$\Stab_\rho=\{g\in G\colon g\rho g^\dagger = \rho\}.
$$
The corresponding local unitary stabilizer subalgebra, which we denote
by $K_\rho$, is the Lie algebra of $\Stab_\rho$, given by
$$K_\rho = \{M\in LG\colon [M,\rho]=0\}.
$$

We will use the standard basis
$$A=i\sigma_z, \; B=i\sigma_y, \; C=i\sigma_x$$
for $su(2)$.
Given an element $M\in su(2)$,
we will write $M^{(k)}$ to denote the element $(0,0,\ldots,M,\ldots,0)$ in $LG$ with
$M$ in the $k$th position and zero matrices in all other positions.

We will make use of the following elementary 1-qubit Lie algebra actions on Pauli
tensors.
\begin{eqnarray}
\left[A,\sigma_x\right] &=& -2\sigma_y\\
\left[A,\sigma_y\right] &=& 2\sigma_x\\
\left[A,\sigma_z\right] &=& 0  
\end{eqnarray}
Less trivial, but still elementary to check, are the following Lie
algebra action calculations in standard coordinates~\cite{maxstabnonprod2}.
\begin{eqnarray}
\left[\left(\sum_k a_kA^{(k)}\right), \rho \right]&=& \sum_{I,J} 
\zeta(I,J) \ket {I}\bra{J} \label{sumAaction}\\ 
\left[\left(\sum_k c_kC^{(k)}\right), \rho\right] &=& \sum_{I,J} 
\eta(I,J) \ket{I}\bra{J} \label{sumCaction}
\end{eqnarray}
where $\zeta(I,J),\eta(I,J)$ are given by
\begin{eqnarray*}
\zeta(I,J) &=&  2i \rho_{I,J} \sum_{i_\ell\neq j_\ell}
(-1)^{i_\ell}a_\ell\\
\eta(I,J) &=& i\sum_{k,\ell=1}^n
(c_k \rho_{I_k,J} - c_\ell\rho_{I,J_\ell}).
\end{eqnarray*}


\section{Symmetric states and polynomials}\label{symmstate2poly}


Given a permutation $\pi$ of $\{1,2,\ldots,n\}$, define $P_\pi$ to be the operator
on Hilbert space that carries out the corresponding permutation of qubits.
For example,
\begin{align*}
P_{(23)} \ket{11000} &= \ket{10100} \\
P_{(23)} \left( \sigma_j \otimes \sigma_k \otimes \sigma_l \right) P_{(23)}^{-1}
 &= \sigma_j \otimes \sigma_l \otimes \sigma_k
\end{align*}
Define a symmetrization operator on density matrices as
\[
\Sym(\rho) = \frac{1}{n!} \sum_{\pi} P_\pi \rho P_\pi^{-1} .
\]
An $n$-qubit symmetric density matrix $\rho$
is one for which $P_\pi \rho P_\pi^{-1} = \rho$ for all permutations $\pi$,
or equivalently, one for which $\Sym(\rho) = \rho$.
We denote by $\Sym^n {\cal W}$ the $n$-fold symmetric power of ${\cal
  W}$ defined in the previous section.  It is the subspace
of elements of ${\cal W}^{\otimes n}$ that are invariant under qubit permutation.
Every $n$-qubit symmetric density matrix $\rho$ is an element
of $\Sym^n {\cal W}$, and can be written
\begin{equation}\label{symrho}
\rho = \frac{1}{2^n} \sum c_{n_1n_2n_3} \Sym \left(
  \sigma_0^{\otimes n_0} \otimes \sigma_1^{\otimes n_1} \otimes
  \sigma_2^{\otimes n_2} \otimes \sigma_3^{\otimes n_3} \right) ,
\end{equation}
where the sum is over non-negative integers $n_0$, $n_1$, $n_2$, and $n_3$
such that $n_0 + n_1 + n_2 + n_3 = n$.
The coefficients $c_{n_1n_2n_3}$ are real.
The collection of $n$-qubit symmetric density matrices is a proper subset of
$\Sym^n {\cal W}$, since the latter contains
Hermitian matrices that are not positive semi-definite, and Hermitian matrices for
which the trace is not 1.

Let $\R_n[x,y,z]$ be the set of polynomials of degree at most $n$
in three variables $x$, $y$, and $z$ with
real coefficients.
For each $n$, there is a linear map $F_n: \Sym^n {\cal W} \to \R_n[x,y,z]$ defined by
\[
\frac{1}{2^n} \Sym \left(
  \sigma_0^{\otimes n_0} \otimes \sigma_1^{\otimes n_1} \otimes
  \sigma_2^{\otimes n_2} \otimes \sigma_3^{\otimes n_3} \right) \mapsto
x^{n_1} y^{n_2} z^{n_3} .
\]
In this way, we may associate a polynomial of degree at most $n$ with each
$n$-qubit symmetric mixed state.  The polynomial associated with
(\ref{symrho}) is
\begin{equation}
F_n(\rho) = \sum_{\stackrel{n_1,n_2,n_3}{n_1+n_2+n_3 \leq n}} c_{n_1n_2n_3} x^{n_1} y^{n_2} z^{n_3} .
\end{equation}
For each $n$, the map $F_n$ is an invertible linear map.


Since $F_n$ is a linear map, the polynomial for a mixture of symmetric mixed states
is the mixture of the polynomials.
\[
F_n(p_1 \rho_1 + p_2 \rho_2) = p_1 F_n(\rho_1) + p_2 F_n(\rho_2)
\]

A product of polynomials $\R_n[x,y,z] \times \R_m[x,y,z] \to \R_{n+m}[x,y,z]$
represents the symmetrized tensor product of states.
\[
F_n(\rho_1) F_m(\rho_2) = F_{n+m}(\Sym(\rho_1 \otimes \rho_2))
\]

Let $g \in SU(2)$.  Define $T_g: \Sym^n {\cal W} \to \Sym^n {\cal W}$ to
be the symmetric transformation of each qubit by $g$.
\[
T_g(\rho) = g^{\otimes n} \rho (g^{\dagger})^{\otimes n}
\]
If we define $R_g: \R_n[x,y,z] \to \R_n[x,y,z]$ to be the transformation on polynomials
defined by
\begin{equation}\label{su2actiononpolys}
R_g(f)(x,y,z) = f \left( (x,y,z) \Phi(g) \right) ,  
\end{equation}
where $f$ is a polynomial, $\Phi: SU(2) \to SO(3)$ is the
homomorphism~\footnote{$\Phi$ is given in a natural way by the adjoint
  action $SU(2)\to SO(su(2))$, so that $\Phi(g)(M) = gMg^\dagger$, and we
  identify $su(2)$ with $\R^3$ by $A\leftrightarrow (1,0,0)$,
  $B\leftrightarrow (0,1,0)$, $C\leftrightarrow
  (0,0,1)$. See~\cite{b-td:rclg}.  } that associates a rotation in
$\R^3$ with each $2 \times 2$ unitary, and $(x,y,z) \Phi(g)$ denotes a
row vector multiplied by a $3 \times 3$ orthogonal matrix, then the
following diagram commutes.
\[
\begin{CD}
\R_n[x,y,z] @>R_g>> \R_n[x,y,z] \\
  @AF_nAA           @AF_nAA \\
\Sym^n {\cal W}  @>T_g>> \Sym^n {\cal W}  \\
\end{CD}
\]

Homogeneous polynomials in three variables $x$, $y$, and $z$ are known
to be reducible representations of $SO(3)$ \cite{sternberg}\footnote{To
  be precise, the cited work considers the representation
  $\C[x,y,z]$. In the case of $SO(3)$, the irreducible submodules of
  this complex representation are in one-to-one correspondence with the
  {\em real} irreducible submodules of $\R[x,y,z]$ via
  complexification. See~\cite[Ch.2 Sec.6]{b-td:rclg}.}.  If $V_l$ is the
irreducible representation of $SO(3)$ with dimension $2l+1$, then the
homogeneous polynomials of degree $p$ in three variables decompose into
irreducible representations as
\begin{equation}
\bigoplus_{j=0}^{\lfloor p/2 \rfloor} V_{p-2j} .
\label{homopolyirred}
\end{equation}

\section{Stabilizer structure for symmetric mixed states}\label{algstructsection}

We begin by considering an arbitrary subalgebra $K$ (not necessarily a
stabilizer subalgebra) of $LG=\bigoplus_{i=1}^n su(2)$. That is, $K$ is
a vector subspace of $LG$ that is also closed under the Lie
bracket. Theorem~\ref{lualgdecomp} below (the proof is given in the
appendix) gives a structure theorem for how $K$ decomposes as a direct
sum of basic building blocks. It is convenient to give names to the
following standard algebra types.


Given a subset ${\cal B}\subset\{1,2,\ldots,n\}$ of qubit labels for a
system of $n$ qubits, let $\Delta_{\cal B}$ denote the subalgebra
$$\displaystyle \Delta_{\cal B} = \left\{ \sum_{b\in {\cal B}} M^{(b)} \colon M\in su(2)\right\}
$$ of $LG$. In particular, note that $\Delta_{\cal B}$ is isomorphic to
$su(2)$. We will call $\Delta_{\cal B}$ the {\em standard $su(2)$ block
  subalgebra} for the subsystem specified by ${\cal B}$, and we will use
the term {\em $su(2)$ block algebra} for any algebra that is isomorphic
to some $\Delta_{\cal B}$ via local unitary group conjugation, that is, any
algebra of the form
$$ (g_1\otimes g_2\otimes \cdots \otimes g_n) \Delta_{\cal B}
(g_1\otimes g_2\otimes \cdots \otimes g_n)^\dagger
$$
where $g_i\in SU(2)$ for $1\leq i\leq n$. We will write $\kwerner$ to
denote $\Delta_{\{1,2,\ldots,n\}} = \{(M,M,\ldots,M)\colon M\in
su(2)\}$, because it is the Lie algebra of
the stabilizer group $\{(g,g,\ldots,g)\colon g\in SU(2)\}$
that defines the Werner states. 

We write $\kproduct$ to denote the $n$-dimensional algebra
$$ \kproduct = \left\{ \sum_{k=1}^n a_k A^{(k)} \colon a_k \in \R
\right\}
$$
and we write $\kdicke$ to denote the $1$-dimensional algebra that is the
real span of the element $\displaystyle \sum_{k=1}^n A^{(k)}$. We use 
these names because $\kproduct$ is the stabilizer of the product state 
$\ket{00\cdots 0}\bra{00\cdots 0}$, and $\kdicke$ is the stabilizer of
the symmetric Dicke states~\cite{symmstatespaper}. 
Finally, let $K_{\rm \small GHZ}$ denote the $(n-1)$-dimensional algebra
$$K_{\rm \small GHZ} = \left\{\sum_{i=1}^n a_k A^{(k)}\colon a_k\in \R, \sum_k a_k = 0\right\}
$$
which is isomorphic to $u(1)^{(n-1)}$. The name of this algebra comes from the
fact that it is the stabilizer of the GHZ state\cite{maxstabnonprod2}.

\begin{theorem}\label{lualgdecomp}
Let $K$ be a subalgebra of $LG=\bigoplus_{i=1}^n su(2)$.  The set
$\{1,2,\ldots,n\}$ decomposes as a disjoint union
$$\{1,2,\ldots,n\} = {\cal B}_1 \cup {\cal B}_2 \cup \cdots \cup {\cal
  B}_p \cup {\cal S} \cup {\cal R}$$
such that the algebra $K$ decomposes
as a direct sum of algebras
$$K=\left(\bigoplus_{i=1}^p B_i\right) \oplus S$$ 
with the following properties:
\begin{enumerate}
\item [(i)] each $B_i$ is an $su(2)$ block subalgebra in qubits ${\cal B}_i$,
\item [(ii)] every element in the algebra $S$ has zero coordinates in
  qubits outside of the qubit set ${\cal S}$,
\item [(iii)] for every $s$ in ${\cal S}$, the set $\{M_s\colon
  \sum_k M_k^{(k)}\in S\}$ is a 1-dimensional subalgebra of $su(2)$, and
\item [(iv)] every element in $K$ has zero coordinates in qubits in
  ${\cal R}$.
\end{enumerate}
\end{theorem}

Applying Theorem~\ref{lualgdecomp} to the stabilizer subalgebra
$K_\rho$, where $\rho$ is symmetric, leads to the following theorem,
whose proof is in the appendix.

\begin{theorem}\label{mixedsymmkrhodecomp}
  Let $\rho$ be a symmetric mixed $n$-qubit state. There is an LU
  equivalent symmetric state $\rho'$ such that the local unitary
  stabilizer subalgebra $K_{\rho'}$ is one of the following.
  \begin{enumerate}
  \item[(a)] all of $LG$
\item[(b)] $\kwerner$
\item[(c)] $\kproduct$
\item[(d)] $\kghz$
\item[(e)] $\kdicke$
\item[(f)] the zero algebra
  \end{enumerate}
\end{theorem}

In the following sections, we identify states that have
the stabilizers given in Theorem~\ref{mixedsymmkrhodecomp}. For
stabilizer types~(c), (d), and~(e), we will make use of the following
Lemma, whose proof is in the appendix.

\begin{lemma}\label{gAgdaggerispmA}
Suppose that $K_\rho$ is a stabilizer of type~(c), (d), or~(e), in
Theorem~\ref{mixedsymmkrhodecomp}, and that 
$$(g_1\otimes g_2\otimes \cdots \otimes g_n)K_\rho (g_1\otimes
g_2\otimes \cdots \otimes g_n)^\dagger = K_\rho
$$
for some $g_1,\ldots,g_n$ in $SU(2)$. Then we have 
$g_kAg_k^\dagger = \pm A$ for $1\leq k\leq n$. It follows that each $g_k$ is
either diagonal or antidiagonal. That is, $g_k$ is either of the form
$g_k=\twotwo{e^{it}}{0}{0}{e^{-it}}$ or
$g_k=\twotwo{0}{-e^{-it}}{e^{it}}{0}$ for some real $t$.
\end{lemma}

\section{Stabilizer is all of $LG$}\label{stabLGsect}

The completely mixed state is stabilized by $LG$. In fact this is the
{\em unique} state that has this stabilizer. One can see this as
follows.

  Since $\C^2$ is an irreducible representation of $SU(2)$, it follows
  that $(\C^2)^{\otimes n}$ is an irreducible representation of
  $G=(SU(2))^n$ (see~\cite{b-td:rclg}, Proposition 4.14). By Schur's lemma, any matrix
  that commutes with every $g\in G$ must be a scalar matrix (\cite{b-td:rclg}, Theorem 1.10).

Alternatively, it is straightforward to verify that if $\rho$ has an
off-diagonal element $\rho_{IJ}$ with $I\neq J$, with say $i_k\neq j_k$,
then $\langle I | [A^{(k)},\rho] | J \rangle =(-1)^{i_k}2i \rho_{IJ}\neq 0$
(use~(\ref{sumAaction})). This rules out
nonzero off-diagonal elements. Now suppose $\rho_{II}\neq \rho_{I_kI_k}$
for some $k$, where $I_k$ is the same as $I$ in all positions but
opposite in position $k$.. Then $\langle I | [C^{(k)},\rho] | J \rangle
=i(\rho_{I_kI_k} - \rho_{II}) \neq 0$ (use~(\ref{sumCaction})). We
conclude that $\rho$ is the completely mixed state.

\section{Stabilizer is $\kwerner$}\label{symmwernersect}

Let $\rho$ be an $n$-qubit symmetric mixed state, and suppose that
$K_\rho=\kwerner$. It follows that, on the group level,
we have $\Stab_\rho\supset\{(g,g,\ldots,g)\colon g\in SU(2)\}$. This is the
class of Werner states.  In~\cite{symmwerner} we prove that any
symmetric Werner state must be of the form
$$\rho = \sum_{k=0}^{\lfloor n/2 \rfloor} c_{k} \Sym \left(
u^{\otimes k} \otimes \Id^{\otimes n-2k}
\right)
$$
where 
$$u=\sigma_x\otimes \sigma_x + \sigma_y\otimes \sigma_y + \sigma_z\otimes
  \sigma_z
$$
for some real coefficients $c_{k,\ell}$. Further, we show that any two states
with distinct choices of coefficients $c_{k,\ell}$ are local unitarily inequivalent.



\section{Stabilizer is $\kproduct$}\label{stabu1nsect}

It is clear that if $\rho = \sum_I s_I
\sigma_I$ where $s_I=0$ if $i_k=1,2$ for every $1\leq k\leq n$, then
$K_\rho$ contains $\kproduct$. As long as there is some
$s_I\neq 0$ for $I\neq (0,0,\ldots,0)$, we can apply~(\ref{sumCaction})
to see that $[C^{(k)},\rho]=0$, so we have $K_\rho =\kproduct$ (the only possible
algebra that contains $\kproduct$ is all of $LG$).

Conversely, suppose $K_\rho = \kproduct$. Write
$\rho$ as
\begin{eqnarray*}
\rho &=& \sum_J a_J \sigma_0\otimes \sigma_J
+ \sum_J b_J \sigma_1\otimes \sigma_J\\
&+& \sum_J c_J \sigma_2\otimes \sigma_J
+ \sum_J d_J \sigma_3\otimes \sigma_J
\end{eqnarray*}
where $J$ is a multiindex of length $n-1$. Then we have
$$0=[A_1,\rho]= \sum_J (-2b_J \sigma_3 + 2c_J \sigma_2)\otimes \sigma_J.
$$
It follows that $b_J=c_J=0$ for all $J$. 
Thus we have established
that $K_\rho = \kproduct$ if and only if
$$ \rho = \sum_{k=0}^n c_k \Sym \left(
\sigma_z^{\otimes k} \otimes \Id^{\otimes (n-k)}
\right)
$$
for some real coefficients $c_k$.

In contrast to the case of symmetric Werner states in the previous
section, LU equivalence classes of states whose stabilizer is
$\kproduct$ may contain more than one state. For
example, let $\rho = \Id/8 + a \sigma_z\otimes \sigma_z\otimes
\sigma_z$. Then $\rho' = (iX)^{\otimes 3}\rho (-iX)^{\otimes 3} = \Id/8
-a \sigma_z\otimes \sigma_z\otimes \sigma_z$ is $LU$ equivalent but not
equal to $\rho$, and has the same stabilizer (here $X$ denotes the Pauli
matrix $\sigma_x$ viewed as an element $U(2)$). We claim that this type
of sign change accounts for the entire LU class of a state for $n\geq 3$
qubits. That is, if $\rho,\rho'$ are LU equivalent $n\geq 3$ qubit
states with stabilizer $K_\rho=K_{\rho'}=\kproduct$, and $\rho$ has corresponding polynomial $F_n(\rho)=\sum_k
d_k z^k$, then either $\rho=\rho'$ or $F_n(\rho') = \sum_k (-1)^kd_k
z^k$.  To see why this is true, suppose we have $g\rho g^\dagger =
\rho'$ for some $g=(g_1,\ldots,g_n)$ in $SU(2)^n$. We show
in~\cite{symmstates2} that for $n\geq 3$, there exists an $h\in SU(2)$
such that $\rho'= T_h(\rho) = h^{\otimes n}\rho (h^\dagger)^{\otimes
  n}$. By Lemma~\ref{gAgdaggerispmA} we have $hAh^\dagger=\pm A$. Thus
we have
$R_h(z)=\pm z$. 

Lemma~\ref{gAgdaggerispmA} admits one additional possibility for 2-qubit
states, and that is to have one of $g_1,g_2$ be diagonal, and the other
antidiagonal. An example is conjugation by $\Id \otimes (iX)$ that takes 
$\rho = \frac{\Id}{4} + a\sigma_3\otimes \sigma_3$ to 
$\rho' = \frac{\Id}{4} - a\sigma_3\otimes \sigma_3$. It
is easy to check that this exhausts all possibilities.


\section{Stabilizer is $K_{\rm \small GHZ}$}\label{stabghzlikesect} 

First, we claim that if 
$K_\rho = K_{\rm \small GHZ}$, 
then $\rho_{I,J}=0$ for $J\neq I,I^c$, where $I^c$ denotes the
bitwise complement of $I$. Indeed, if $J\neq I,I^c$, then there exist
qubit indices $k,\ell$ such that $i_k=j_k$ and $i_\ell = j_\ell^c$. Then
we have (apply~(\ref{sumAaction}))
$$0=\langle I | [A^{(k)} - A^{(\ell)},\rho]|J\rangle = \zeta(I,J) = \pm 2i\rho_{I,J}.
$$



Second, we claim that if $\rho_{I,I^c}\neq 0$, then
$\{I,I^c\}=\{00\cdots 0,11\cdots 1\}$. Simply note that for any pair of 
  positions $k,\ell$, we must have 
$$
    0 = (-1)^{i_k} - (-1)^{i_\ell}
$$
so $i_k=i_\ell$.

Let $\ket{\psi(n,k)}=\ket{1}^{\otimes k}\otimes \ket{0}^{\otimes (n-k)}$
and let $\rho({n,k})=\Sym \ket{\psi(n,k)}\bra{\psi(n,k)}$. Let
$\ket{\GHZ({\alpha,\beta})}=\alpha \ket{0}^{\otimes n}+ \beta \ket{1}^{\otimes
  n}$, and let
$\rho_\GHZ(\alpha,\beta)=\ket{\GHZ({\alpha,\beta})}\bra{\GHZ({\alpha,\beta})}$. 
The above
two claims imply that if 
$K_\rho = \kghz$
then $\rho$ can be written as a mixture of the form
\begin{equation}\label{mixedghzform}
\rho = \sum_{k=0}^n c_k\rho({n,k}) + d \rho_\GHZ(\alpha,\beta)
\end{equation}
where $c_0,\ldots,c_n,d$ are nonnegative real numbers that sum to 1.

Now suppose that $\rho,\rho'$ are are LU equivalent $n$-qubit states of
the form~(\ref{mixedghzform}) for some $n\geq 3$. As was the case in the
previous section, there is some $g\in SU(2)$ such that $\rho'=
T_g(\rho)$, and $g$ is either diagonal or antidiagonal. If
$g=\twotwo{e^{it}}{0}{0}{e^{-it}}$, then conjugation by $g^{\otimes n}$
leaves $\rho(n,k)$ fixed and performs a phase operation on
$\rho_\GHZ(\alpha,\beta)$ as follows.
\begin{eqnarray*}
T_g(\rho({n,k})) &=& \rho({n,k})\\
T_g(\rho_\GHZ(\alpha,\beta)) &=& \rho_\GHZ(e^{itn}\alpha,e^{-itn}\beta)
\end{eqnarray*}
If $g=\twotwo{0}{e^{-it}}{e^{it}}{0}$, then $T_g$ interchanges
$\rho(n,k)$ with $\rho(n,n-k)$ and transforms $\rho_\GHZ(\alpha,\beta)$
as follows.
\begin{eqnarray*}
T_g(\rho({n,k})) &=& \rho({n,n-k})\\
T_g(\rho_\GHZ(\alpha,\beta)) &=& \rho_\GHZ(e^{-itn}\beta,e^{itn}\alpha)
\end{eqnarray*}

The considerations of the preceding paragraph imply that if we exercise
LU freedom to choose $\alpha,\beta$ real with $\alpha\geq \beta$,
then~(\ref{mixedghzform}) is a unique representative of its LU class of
states, with the only exception being the case when $\alpha=\beta$.
In this case, $\rho=\sum_k c_k\rho(n,k) +
d\rho_\GHZ(1\sqrt{2},1\sqrt{2})$ and $\rho'=\sum_k c'_k\rho(n,k) +
d\rho_\GHZ(1\sqrt{2},1\sqrt{2})$ are interchanged by $T_{iX}$.

\section{Stabilizer is $\kdicke$}\label{stabu1sect}

Let $\rho$ be an $n$-qubit symmetric mixed state.
We claim that $g^{\otimes n} \rho (g^{\dagger})^{\otimes n} = \rho$
for all \emph{diagonal} $g \in SU(2)$
if and only if
$F_n(\rho)$ is a linear combination of products of
$z^r$ and $(x^2 + y^2)^s$, for $0 \leq r + 2s \leq n$.

As for the case of the single $su(2)$ block, we make use of the
identification of symmetric states with polynomials given in
section~\ref{symmstate2poly}.

If $g$ is a diagonal element of $SU(2)$, then $\Phi(g)$
is a rotation about the $z$~axis.
Polynomials $F_n(\rho) = \sum_{r,s} b_{r,s} z^r (x^2 + y^2)^s$
(where the sum is over nonnegative integers $r$, $s$ such that
$r + 2s \leq n$)
are invariant since $z$ and $x^2 + y^2$ are invariant under rotations
about the $z$~axis.

Conversely, suppose
that $g^{\otimes n} \rho (g^{\dagger})^{\otimes n} = \rho$
for all diagonal $g \in SU(2)$.  Then $F_n(\rho)$ is
invariant under rotations about the $z$~axis.
We now wish to decompose the homogeneous polynomials
in three variables into a sum of irreducible representations
of $U(1)$, and to note the dimension of the trivial
representation, as it gives the space of polynomials
invariant under rotations about the $z$~axis.
Each $V_l$ in (\ref{homopolyirred}) decomposes
into $2l+1$ one-dimensional irreducible representations of $U(1)$,
of which one is the trivial representation.
When expressed as a sum of irreducible
representations of $U(1)$, decomposition (\ref{homopolyirred}) shows that
the homogeneous polynomials of degree $p$
in three variables contain $\lfloor p/2 \rfloor + 1$ dimensions of the trivial
representation.  This is precisely the dimension
of the space of homogeneous polynomials of degree $p$
that can be produced by products of $z^r$ and $(x^2 + y^2)^s$
(\emph{i.e.} the number of nonnegative integer pairs $(r,s)$
for which $r + 2s = p$).
The vector space $\R_n[x,y,z]$ of polynomials
of degree at most $n$ is a direct sum of vector spaces of
homogeneous polynomials of degree $p$, for $0 \leq p \leq n$.
Since we have accounted for the dimensions of each space
of homogeneous polynomials, $F_n(\rho)$
must be a linear combination of the polynomials given.

Now suppose that $\rho,\rho'$ are LU equivalent, with $F_n(\rho)
=\sum_{r,s} b_{r,s} z^r (x^2 + y^2)^s$, and $F_n(\rho')=\sum_{r,s}
b'_{r,s} z^r (x^2 + y^2)^s$.  For $n\geq 3$ qubits, the same
analysis as for symmetric Werner states (Section~\ref{symmwernersect})
yields a diagonal or antidiagonal $g\in SU(2)$ such that
$T_g(\rho)=\rho'$. Since $R_g(x^2+y^2) = x^2+y^2$ and $R_g(z)=\pm z$,
the only nontrivial possibility is that $F_n(\rho')=\sum_{r,s}
(-1)^r b_{r,s} z^r (x^2 + y^2)^s$. We can discard the case $n=2$, because
the only possibilities for $\rho$ are $\rho = \Id/4 + a (\sigma_x\otimes
\sigma_x + \sigma_y\otimes \sigma_y)$ and $\rho = \Id/4 + a
\sigma_z\otimes \sigma_z$. The first of these states has an $su(2)$
block stabilizer, and the second has stabilizer $\kproduct$.

\section{Summary and Conclusion}

Table~\ref{tab:summarytable} shows a summary of LU classes of local
unitary stabilizer algebra types and their corresponding LU classes of
states. Having achieved LU classification for symmetric mixed states, it
is natural to attempt further classes of mixed states. A natural avenue
for investigation is to take the stabilizer structures from
Theorem~\ref{mixedsymmkrhodecomp} and try to classify their
corresponding states, with the assumption of permutation invariance
removed. For example, the case Werner states (stabilizer type~(b)) would
be of interest.

  \begin{center}
\begin{table*}
  \begin{tabular}{c|c|c|c}
{\bf Stabilizer} & {\bf LU Representative} & {\bf LU Nonuniqueness} & {\bf Pure State Example}\\ \hline 
$LG$ & $\Id/2^n$ & unique & none\\ \hline $\kwerner$ &
\raisebox{-3ex}{\shortstack{$\displaystyle \sum_{k=0}^{\lfloor n/2 \rfloor} c_{k}
    \Sym\left((\sigma_x\otimes \sigma_x + \sigma_y\otimes \sigma_y +
    \sigma_z\otimes \sigma_z)^{\otimes k} \otimes \Id^{\otimes (n-2k}\right)$\\ $c_{k}$ real, at least one $c_{k}\neq 0$
    with $k>0$ }} & unique & singlet \\ \hline $\kproduct$ &
\raisebox{-3ex}{\shortstack{$\displaystyle \sum_{k=0}^n c_{k}
    \Sym\left(\sigma_z^{\otimes k} \otimes \Id^{\otimes (n-k)}\right)$
    \\ $c_k$ real, at least one $c_k$ nonzero with $k>0$ }} & $c_k'=(-1)^kc_k$ &
product state\\ \hline $\kghz$ &
\raisebox{-3ex}{\shortstack[c]{$\displaystyle \sum_{k=0}^n c_{k}
    \rho(n,k) + d\rho_{\GHZ}G(\alpha,\beta)$\\ $\alpha\geq \beta>0$ }} &
$c_k'=c_{n-k}$ & GHZ\\ \hline $\kdicke$ &
\raisebox{-3ex}{\shortstack[c]{$\displaystyle \sum_{2r+s\leq n} b_{r,s}
    \Sym\left((\sigma_x\otimes \sigma_x + \sigma_y\otimes
    \sigma_y)^{\otimes s}\otimes \sigma_z^{\otimes r} \otimes \Id^{\otimes n-r-2s}
    \right)$\\ $b_{r,s}$ real, at least one $b_{r,s}\neq 0$ with $r>0$
}} & $b_{r,s}'=(-1)^rb_{r,s}$ & Dicke state
  \end{tabular}
  \caption{\label{tab:summarytable} Summary of LU classes of stabilizers
    and states. Each of the three last stabilizer types has precisely
    two LU inequivalent states of the form given in the column `LU
    Representative'. The column `LU Nonuniqueness' gives conditions on
    the coefficients for the alternative state.}
\end{table*}
  \end{center}

\medskip

{\em Acknowledgments.}  
This work has been supported by National Science Foundation grant
\#PHY-0903690.

\appendix

\section{Proofs for Section~\ref{algstructsection}}

\subsection{Proof of Theorem~\ref{lualgdecomp}}

Let $K$ be a subalgebra of $LG=\bigoplus_{i=1}^n su(2)$.
For each qubit $i$, let $\pi_i\colon LG \to su(2)$ denote the projection
$(M_1,M_2,\ldots, M_i,\ldots,M_n)\to M_i$ onto the $i$th direct summand
of $LG$. Given a set ${\cal S}$ of qubits, let $\pi_{\cal S}$ denote
the projection $\pi_{\cal S} = \oplus_{s\in {\cal S}} \pi_s$ onto
summands in qubits ${\cal S}$. Given $M=(M_1,\ldots,M_n)$ in $LG$, let the
{\em weight of $M$}, denoted $\wt(M)$, be the number of $i$ such that
$M_i$ is not zero.

For each $i\in\{1,2,\ldots,n\}$, let $m(i)$ be the minimum weight of all
elements of $K$ with nonzero weight in position $i$. For $i$ such that
$m(i)> 0$, let
$M(i)\subseteq K$ denote the set of elements of $K$ that have nonzero
projection in position $i$ and have total weight $m(i)$.  That is,
\begin{eqnarray*}
m(i) &=&\min\{\wt(M)\colon M \in K\colon M_i\neq 0\}, \mbox{ and}  \\
M(i) &=& \{M\in K\colon M_i\neq 0 \mbox{ and } \wt(M)=m(i)\}.
\end{eqnarray*}

\begin{proposition}\label{miproj3samenonzeroposn}
Suppose $\dim \pi_i(K)=3$, and let $P,Q$ be elements of $M(i)$.  For
$1\leq j\leq n$ we have $P_j=0$ if and only if $Q_j=0$.
\end{proposition}

\begin{proof}
Let $U'$ be an element of $K$ such that $U'_i,P_i$ are independent
elements of $su(2)$.  Let $V=P/|P_i|$, let $W=[U',V]/|[U',V]_i|$, and
let $U=[V,W]$.

Observation 1. We have that $U_i,V_i,W_i$ form an
orthonormal~\footnote{We use a rescaled Hilbert-Schmidt norm
  $|M|=(2\tr(M^\dagger M))^{1/2}$ for an element $M$ of $su(2)$.} basis of $su(2)$ that
satisfy $[U_i,V_i]=W_i$, $[V_i,W_i]=U_i$, and $[W_i,U_i]=V_i$.  This is
easy to see using the fact that $su(2)$ is isomorphic to $\R^3$, with
the Lie bracket operation corresponding to the cross product, and the
(rescaled) Hilbert-Schmidt norm corresponding to the standard norm on $\R^3$.

Observation 2. We have that $U,V,W$ are elements of $M(i)$.  This
follows from the observation that in $K$, we have $\wt([M,N]) \leq
\wt(M)$ for all $M,N$, so $U,V,W$ all have weight less than or equal to
$\wt(P)$, and hence must have weight equal to $\wt(P)=m(i)$.  Further,
we see that $U,V,W$ must have zero and nonzero coordinates in the same
positions as $P$.

By Observation~1, we have that $Q_i$ is not a scalar multiple of at
least one of $U_i,V_i,W_i$.  Without loss of generality,  suppose
$Q_i,U_i$ are linearly independent.  Let $S=[Q,U]$, so by construction,
we have $S_i\neq 0$, $S\in M(i)$, and $S$ has zero and nonzero
coordinates in the same positions as $U$, and hence in the same
positions as $P$.  If $Q_j\neq 0$, then $S_j$ must be nonzero since
$\wt(S)=\wt(Q)=m(i)$ and $S$ has the same zero and nonzero coordinate
positions as $P$, which also has weight $m(i)$.  Thus $U_j\neq 0$, and
therefore $P_j\neq 0$.  Conversely, if $P_j\neq 0$, then $S_j\neq 0$,
and therefore $Q_j\neq 0$.  This concludes the proof.
\end{proof}

\begin{proposition}\label{su2blockproperties} 
Suppose $\dim \pi_i(K)=3$.  Let
  ${\cal B}\subseteq \{1,2,\ldots,n\}$ be the set of positions where all
  the elements of $M(i)$ have their nonzero coordinates, by virtue
  of Proposition~\ref{miproj3samenonzeroposn}.  Then the following hold.
  \begin{enumerate}
  \item [(i)] If $P\in K$ and $P_j=0$ for all
  $j\not\in {\cal B}$, then $P\in M(i)$ or $P=0$.  
\item[(ii)] Any such $P$
  lies in the linear span of $U,V,W$ constructed in the proof
  of Proposition~\ref{miproj3samenonzeroposn}.
\item[(iii)] The elements $U,V,W$ satisfy $[U,V]=W$, $[V,W]=U$,
  $[W,U]=V$ so that the linear span of $U,V,W$ is isomorphic to $su(2)$.
\item[(iv)] $\displaystyle K=\pi_{\cal B}(K) \oplus \pi_{{\cal
    B}^c}(K).$
\item[(v)] $\pi_{\cal B}(K)$ is an $su(2)$ block algebra.
  \end{enumerate}
\end{proposition}

\begin{proof}
(i). If $P=0$ or if $P_i\neq 0$ we are done, so suppose $P_i=0$, and let $j\neq i$ be an
 element of ${\cal B}$ such that $P_j\neq 0$.   Let $U,V,W$ be constructed as in the proof
  of Proposition~\ref{miproj3samenonzeroposn}.  Since $U_j,V_j,W_j$ are nonzero,
  they are independent elements of $su(2)$.  Write $P_j=aU_j+
  bV_j+cW_j$, and let $Q=P-aU-bV-cW$.  Then $Q_i=-aU_i-bV_i-cW_i\neq 0$,
  so $Q\in M(i)$.  But $Q_j=0$, and this
  contradicts Proposition~\ref{miproj3samenonzeroposn}.  We conclude that $P_i=0$
  is impossible, so $P\in M(i)$.

(ii). Using (i), for $P\neq 0$ we can write $P_i=aU_i+bV_i+cW_i$ and let $Q=P-aU-bV-cW$.  Then
$Q_i=0$,  so $Q=0$, again by~(i).  Thus $P=aU+bV+cW$.

(iii).  We only need to show that $[U_k,V_k]=W_k$, $[V_k,W_k]=U_k$, and
$[W_k,U_k]=V_k$ for $k\in {\cal B}$.  By construction, we have
    $[V,W]=U$, and therefore we have $[V_k,W_k]=U_k$.  We also have
    $[W_k,U_k]=\alpha_k V_k$ and $[U_k,V_k]=\beta_k W_k$ for some
    scalars $\alpha_k,\beta_k$.  Let $N=[V,[U,V]]$, and let $M=N-U$.  Then
    $N_k=[V_k,\beta_k W_k]=\beta_k U_k$.  Since $M_i=0$, we must have
    $M_k=(\beta_k-1)U_k=0$ by~(i), so we must have $\beta_k=1$ for all
    $k\in{\cal B}$.  A similar argument shows that $\alpha_k=1$ for all
    $k\in{\cal B}$.

(iv). Let $M\in A$, and write $M=P+N$, where $P=\pi_{\cal B}(M)$ and
  $N=M-P$.  We will show that $P\in A$.  We have $[M,X]=[P,X]$ for
  $X=U,V,W$, so we can use~(ii) to write $[P,X]$ as a linear combination
  of $U,V,W$.
  \begin{eqnarray*}
    \left[P,U\right] &=& a_UU+ b_UV + c_UW\\
    \left[P,V\right] &=& a_VU+ b_VV + c_VW\\
    \left[P,W\right] &=& a_WU+ b_WV + c_WW
  \end{eqnarray*}
For $k\in {\cal B}$, we can write $P_k=a_kU_k+b_kV_k+c_kW_k$.  Using the
relations~(iii), we have
  \begin{eqnarray*}
    \left[P_k,U_k\right] &=& -b_kW_k+c_kV_k\\
    \left[P_k,V_k\right] &=& a_kW_k-c_kU_k\\
    \left[P_k,W_k\right] &=& -a_kV_k+b_kU_k.
  \end{eqnarray*}
Equating $k$th coefficients of the first set of equations with
coefficients in the second set of equations yields
\begin{eqnarray*}
  a_k &=& -b_W = c_V \\
  b_k &=& -c_U = a_W\\
  c_k &=& = -a_V = b_U.
\end{eqnarray*}
Since this holds for all $k$, we have $P=c_V U + a_W V + b_U W$, and so
$P\in K$, as desired.

(v). Let $b\in {\cal B}$. Because $U_b,V_b,W_b$ satisfy $[U_b,V_b]=W_b$,
$[V_b,W_b]=U_b$, and $[W_b,U_b]=V_b$, it must be that $U_b,V_b,W_b$ form
an orthonormal basis for $su(2)$ (again use the fact that $su(2)$ with
its Lie bracket and Hilbert-Schmidt norm is isometrically identified
with $\R^3$ with the cross product and standard euclidean norm). Because
the adjoint representation $SU(2)\to SO(su(2))$ is surjective, we may
choose $h_b\in SU(2)$ such that $h_bU_bh_b^\dagger = A$. Now choose a
real number $t_b$ such that $e^{it_b}V_be^{-it_b}=B$. Now let
$g_b=e^{it_b}h_b$. For $b'\not\in{\cal B}$, let $g_{b'}=\Id$. We now
have
\begin{eqnarray*}
(\otimes_i g_i) U (\otimes_i g_i)^\dagger &=& (A,A,\ldots,A)\\
(\otimes_i g_i) V (\otimes_i g_i)^\dagger &=& (B,B,\ldots,B)\\
(\otimes_i g_i) W (\otimes_i g_i)^\dagger &=& (C,C,\ldots,C)
\end{eqnarray*}
and $\pi_{\cal B}(K)$ has been ``aligned'' with $\Delta_{\cal B}$, as desired.
\end{proof}

To complete the proof of Theorem~\ref{lualgdecomp}, 
for each qubit label $i$ for which $\dim \pi_i(K)=3$,
let ${\cal B}\subseteq \{1,2,\ldots,n\}$ be the set of positions where
all the elements of $M(i)$ have their nonzero coordinates, as in
Propositions~\ref{miproj3samenonzeroposn}
and~\ref{su2blockproperties}. Then $\pi_{\cal B}(K)$ is an su(2) block
summand for $K$ by part~(v) of
Proposition~\ref{su2blockproperties}. For any qubit $j$ outside of the qubit
sets for $su(2)$ blocks, we must have $\dim \pi_j(K)=0,1$, for if there
are two independent vectors in $\pi_j(K)$, then there must also be a
third coming from the bracket of the two independent elements, and so
$j$ would be a qubit for an $su(2)$ block. We define
${\cal S}$ to be the set of qubits $j$ for which $\dim \pi_j(K)=1$ and
define ${\cal R}$ to be the remaining qubits $\ell$, where we have $\dim
\pi_\ell (K)=0$. With these definitions, we clearly have the desired
decomposition of Theorem~\ref{lualgdecomp}.

\subsection{Proof of Theorem~\ref{mixedsymmkrhodecomp}}

Let $\rho$ be a symmetric mixed state, and let $K_\rho$ be its local
unitary stabilizer subalgebra. By Theorem~\ref{lualgdecomp}, we can
decompose $K_\rho$ as a direct sum
\begin{equation}\label{krhodecomp}
K_\rho=\left(\bigoplus_{i=1}^p B_i\right) \oplus S
\end{equation}
of $su(2)$ blocks $B_i$ and an algebra $S$ where projections into each
qubit summand are 1-dimensional.

To begin, note that permutation
invariance implies that the dimension of the projection of $K_\rho$ into
any $su(2)$ summand of $LG$ must be the same for all qubits.  Thus we
can have only the zero algebra (possibility~(f) of
Theorem~\ref{mixedsymmkrhodecomp}), an algebra $S$, or a sum of $su(2)$
blocks.  

Next, we claim that there are only two possibilities for a sum of
$su(2)$ blocks. One extreme is to have $n$ $su(2)$ blocks, each in 1
qubit (possibility~(a) of Theorem~\ref{mixedsymmkrhodecomp}). The other
is to have a single $su(2)$ block in all $n$ qubits. We rule out the
intermediate possibilities, that is, having two or more $su(2)$ blocks,
at least one of which involves two or more qubits, as follows. Suppose
${\cal B}_1$ contains qubits $i$ and $j$, with corresponding $su(2)$
block algebra $B_1$, and ${\cal B}_2$ contains qubit $k\neq i,j$, with
corresponding $su(2)$ block algebra $B_2$. Transposing qubits $j,k$ does
not affect $\rho$, so $K_\rho$ also has an $su(2)$ block algebra $B_3$
that contains qubits $i,k$. But then the qubit sets for $B_1$ and $B_3$
both contain $i$. This contradicts the fact that qubit sets for $su(2)$
blocks are disjoint, as shown in Proposition~\ref{su2blockproperties}.

To complete part~(b) of Theorem~\ref{mixedsymmkrhodecomp}, we consider
two cases.

{\em Part (b) case~(i) $n\geq 3$.} We claim that if $K_\rho$ is a single
$su(2)$ block in all $n\geq 3$ qubits, then in fact $K_\rho$ is the
standard $su(2)$ block algebra $\Delta_{\{1,2,\ldots,n\}}$. Suppose on
the contrary that there is an element $M=\sum_i M_i^{(i)}$ in $K_\rho$
with $M_i\neq M_j$. Let $M'=M_j^{(i)} + M_i^{(j)} + \sum_{k\neq
  i,j}M_k^{(k)}$ be the element obtained from $M$ by transposing the
$i,j$ coordinates. Then $M-M'=\left( M_{i} - M_{j}\right)^{(i)} + \left(
M_{j} -M_{i}\right)^{(j)}$ is also in $K_\rho$, but this element has
weight $2<n$. This contradicts Proposition~\ref{su2blockproperties} that
says all elements in an $su(2)$ block must have full weight $n$. We
conclude that a single $su(2)$ block stabilizer for a symmetric mixed
state of $n\geq 3$ qubits must be the standard $su(2)$ block algebra.

{\em Part (b) case~(ii) $n=2$.} We claim that any two-qubit symmetric
mixed state with an $su(2)$-block stabilizer is LU equivalent to another
symmetric state with a standard $su(2)$-block stabilizer.  This LU
equivalence may not be achievable through the same unitary operation on
each qubit.  For example, the pure state $\ket{01} + \ket{10}$ is a
symmetric state with a non-standard $su(2)$-block stabilizer.  It is LU
equivalent to the singlet state $\ket{01} - \ket{10}$, but not through
an LU transformation applied uniformly to each qubit.  In fact, the
singlet is invariant under any LU transformation applied uniformly to
each qubit.

Suppose $\rho$ is a 2-qubit symmetric mixed state whose stabilizer
subalgebra is an
$su(2)$ block subalgebra.  Writing
\[
\rho =  \sum_{i,j=0}^3 s_{ij} \sigma_{i} \otimes
\sigma_{j} ,
\]
we see that
\[
s_{10} = s_{20} = s_{30} = s_{01} = s_{02} = s_{03} = 0
\]
because the three components $(s_{10},s_{20},s_{30})$ transform like a
3-dimensional real vector under the rotations of 3-dimensional space
produced by the first qubit in the $su(2)$ block subalgebra.  Since
$\rho$ is invariant under the action of this $su(2)$ block subalgebra,
it must be that $(s_{10},s_{20},s_{30}) = 0$.  A similar argument holds
for the second qubit.  By performing singular value decomposition on the
real $3 \times 3$ matrix of elements $s_{i_1i_2}$ with $i_1$ and $i_2$
ranging over $\{1,2,3\}$, we can find a local unitary transformation
such that $\rho' = (g_1 \otimes g_2) \rho (g_1 \otimes g_2)^\dagger$ has
the form
\[
\rho' = \frac{1}{4} \left( \sigma_0 \otimes \sigma_0
 + a \sigma_1 \otimes \sigma_1
 + b \sigma_2 \otimes \sigma_2
 + c \sigma_3 \otimes \sigma_3 \right) .
\]

Let
\[
M = \alpha_1 A^{(1)} + \beta_1 B^{(1)} + \gamma_1 C^{(1)}
 + \alpha_2 A^{(2)} + \beta_2 B^{(2)} + \gamma_2 C^{(2)} .
\]
This $M$ stabilizes $\rho'$ if and only if
\begin{align*}
\alpha_1 c &= \alpha_2 b \\
\alpha_1 b &= \alpha_2 c \\
\beta_1 a &= \beta_2 c \\
\beta_1 c &= \beta_2 a \\
\gamma_1 b &= \gamma_2 a \\
\gamma_1 a &= \gamma_2 b .
\end{align*}
Since there is, by assumption, an element $M$ with $\alpha_1$,
$\beta_1$, $\gamma_1$,
$\alpha_2$, $\beta_2$, and $\gamma_2$ all nonzero, we must have
\[
\abs{a} = \abs{b} = \abs{c}.
\]

If $a$, $b$, and $c$ are not identical, then two have the same
sign and the last has the opposite sign, and we can do a
unitary transformation on one qubit that produces a $\pi$
rotation about the $x$, $y$, or $z$ axis, resulting
in an LU equivalent state that is symmetric with the
standard $su(2)$ block subalgebra as its stabilizer. This completes
part~(b), case~(ii).

Now we suppose that $K_\rho=S$, so that the projection of $K_\rho$
into each $LG$ summand $su(2)$ is 1-dimensional.  We wish to show that
one of possibilities~(c)--(e) of the statement of the Theorem holds. We
do this by cases.

If $K_\rho$ contains a weight one element $M^{(k)}$, then by permutation
invariance, $K_\rho$ must contain $M^{(j)}$ for $1\leq k\leq n$. This
establishes possibility~(c).

Now suppose that $K_{\rho}$ contains no weight one elements. There are
two cases: there may exist a nonzero element $M=(M_{1},M_2,\ldots,M_n)=
\sum_k M_k^{(k)}$ in $K_\rho$ such that $M_i\neq M_j$ for some $i,j$, or
there may not. Suppose the first case holds.  For any pair $i,j$ of
qubit positions, we have the element 
$M'=M_j^{(i)} + M_i^{(j)} + \sum_{k\neq i,j}M_k^{(k)}$, 
obtained by transposing coordinates $i,j$,
also in $K_\rho$.  Therefore we have 
$M-M'=\left( M_{i} - M_{j}\right)^{(i)} + \left( M_{j} -M_{i}\right)^{(j)}$ 
also in
$K_\rho$.  Let $N=M_i-M_j$.  We have $n-1$ linearly independent weight
two elements $N^{(1)}-N^{(k)}$, $2\leq k\leq n$, so the dimension of
$K_\rho$ is at least $n-1$.  If the dimension of $K_{\rho}$ is greater
than $n-1$, then there is an element in $K_{\rho}$ of the form
$M''=\sum_{k=1}^{n} a_{k} N^{(k)}$ that does not lie in the span of the
elements $N^{(i)}-N^{(j)}$.  But then we would have the weight one
element $\left(\sum_{k=1}^n a_k\right)N^{(k)} = M''+\sum_{k=2}^{n} a_{i}
\left(N^{(1)}-N^{(k)} \right)$ in $K_{\rho}$, which contradicts our
assumption. We conclude that $K_{\rho}$ is the real span of
$\{N^{(1)}-N^{(i)}\}_{i=2}^n$. We may take $N$ to have norm 1, and we may choose $g\in SU(2)$
such that $gNg^\dagger = A$. Let $\rho'=g^{\otimes n} \rho (g^{\otimes
  n})^\dagger$. Then 
$K_{\rho'}$ is the real span of $\{A^{(i)}-A^{(j)}\}_{1\leq
  i,j\leq n}$, which is the stabilizer~(d) of Theorem~\ref{mixedsymmkrhodecomp}.

The last case to consider is where $M_{i} = M_{j}$ for all $i,j$, for
all $M=(M_1,M_2,\ldots,M_n)$ in $K_\rho=S$.  Let $M=(N,\ldots,N)$ be a
nonzero element in $K_\rho$.  Normalize $N$ and choose $g\in SU(2)$ to
diagonalize $N$ so we have $gNg^\dagger = A$.  Then the state
${\rho'}=g^{\otimes n}\rho(g^\dagger)^{\otimes n}$ has
1-dimensional stabilizer that is the real span of
$A^{(1)}+A^{(2)}+\cdots +A^{(n)}$. This is type~(e) in
Theorem~\ref{mixedsymmkrhodecomp}.

This concludes the proof of Theorem~\ref{mixedsymmkrhodecomp}.

\subsection{Proof of Lemma~\ref{gAgdaggerispmA}}
 
Suppose that $K_\rho$ is a stabilizer of type~(c), (d), or~(e), in
Theorem~\ref{mixedsymmkrhodecomp}, and that
$$(g_1\otimes g_2\otimes \cdots \otimes g_n)K_\rho (g_1\otimes
g_2\otimes \cdots \otimes g_n)^\dagger = K_\rho
$$
for some $g_1,\ldots,g_n$ in $SU(2)$.
For $M$ in any of the three stabilizer types, the assumption that 
$(\otimes_k g_k)M(\otimes_k(g_k^\dagger) \in K_\rho$ implies that
$g_kAg_k^\dagger = \pm A$ for each $k$. It is easy to check directly that if
$g=\twotwo{a}{-\overline{b}}{b}{\overline{a}}$ satisfies $gAg^\dagger =
\pm A$, then either $a$ or $b$ must be zero. Thus we conclude that $g$ is
either diagonal or antidiagonal, as desired.

\end{document}